%% file: paper.tex
\documentclass{llncs}
\usepackage{amsfonts,amsmath}
\usepackage{tikz}
\usepackage{fullpage}
\title{Scheduling Two Agents on a Single Machine:\\ A Parameterized Analysis of NP-hard Problems\thanks{Preliminary version of this work appeared in the proceedings of the International symposium on Parameterized and Exact Computation (IPEC) 2015.}}

\author{Danny Hermelin\inst{1} \and Judith-Madeleine Kubitza\inst{2} \and Dvir Shabtay\inst{1} \and \\ Nimrod Talmon\inst{3} \and Gerhard Woeginger\inst{4}}

\institute{
Ben Gurion University of the Negev, Israel \\
\email{hermelin@bgu.ac.il, dvirs@bgu.ac.il} \and
TU Berlin, Germany \\
\email{judith-madeleine.kubitza@campus.tu-berlin.de} \and
Weizmann Institute of Science, Israel \\
\email{nimrodtalmon77@gmail.com} \and
Eindhoven University of Technology, The Netherlands \\
\email{gwoegi@win.tue.nl}
}

\begin{document}

\maketitle
\begin{abstract}
\input{abstract}
\end{abstract}

\section{Introduction}
\input{introduction}

\section{Preliminaries}
\label{section: pre}
\input{preliminaries}

\section{Weighted Sum of Completion Times}
\label{section: TotalWeightedCompletionTimes}
\input{completion}

\section{Weighted Number of Tardy Jobs}
\label{section:WeightedNumberTardyJobs}
\input{tardy}

\section{Weighted Number of Just-in-Time Jobs}
\label{section: WeightedNumberJustinTimJobs}
\input{jit}

\section{Conclusions and Open Problems}
\label{section: Conclusion}
\input{conclusion}

\section{Acknowledgments}
\input{Acknowledgments}

\end{document}

%% file: abstract.tex

Scheduling theory is an old and well-established area in combinatorial optimization, whereas the much younger area of parameterized complexity has only recently gained the attention of the community.  Our aim is to bring these two areas closer together by studying the parameterized complexity of a class of single-machine two-agent scheduling problems. Our analysis focuses on the case where the number of jobs belonging to the second agent is considerably smaller than the number of jobs belonging to the first agent, and thus can be considered as a fixed parameter $k$.  We study a variety of combinations of scheduling criteria for the two agents, and for each such combination we pinpoint its parameterized complexity with respect to the parameter $k$. The scheduling criteria that we analyze include the total weighted completion time, the total weighted number of tardy jobs, and the total weighted number of just-in-time jobs.  Our analysis draws a borderline between tractable and intractable variants of these problems. 

%% file: introduction.tex

Scheduling is a well-studied area in operations research that provides fertile grounds for several combinatorial problems. In a typical scheduling problem, we are given a set of jobs that are to be scheduled on a set of machines which is arranged according to a specific machine setting. The objective is to determine a schedule which minimizes a predefined scheduling criterion such as the makespan, total weighted completion time, and total weighted tardiness of the schedule. There are various machine settings including the single machine setting, parallel machines, flow-shop and job-shop, and each scheduling problem may in addition have various attributes and constraints. We refer the reader to \emph{e.g.}~\cite{Brucker2006,Leung2004,Pinedo2008} for an extensive introduction to the area of scheduling, and for a detailed survey of classical results.

Many scheduling problems are NP-hard. Typically, such hard problems include a multitude of parameters, and many NP-hardness proofs exploit the fact that these parameters can be arbitrary large in theory. However, in many practical settings, one or more of these parameters will actually be quite small. For example, the number of different items that can be processed in the shop might be limited, resulting in a scheduling instance where only a limited number of different processing times appear. A limited set of planned delivery dates resulting in a scheduling instance with a limited number of different due dates is another example. It is therefore natural to ask whether NP-hard scheduling problems become tractable when some of their parameters can be assumed to be comparatively small in practice. Luckily, a framework for answering such questions has been recently developed by the computer science community - the theory of parameterized complexity.

Parameterized complexity facilitates the analysis of computational problems in terms of various instance parameters that may be independent of the total input length. In this way, problem instances are analyzed not only according to the total input length $n$, but also according to an additional numerical parameter $k$ that may encode other aspects of the input. A problem is considered tractable if there is an algorithm that optimally solves any instance in $f(k) \cdot n^{O(1)}$ time, where $f()$ is allowed to be any arbitrary computable function which is independent of $n$, and the exponent in $n^{O(1)}$ is required to be independent of $k$. For example, a running-time of $2^{O(k)} \cdot n ^3$ is considered tractable in the parameterized setting, while $n^{O(k)}$ is not. In this way we can model scenarios where certain problem parameters are typically much smaller than the total input length, yet may not be small enough to be considered constant. 

Parameterized complexity has enjoyed tremendous success since its first developments in the early 90s, as can be exemplified by the various textbooks on the subjects~\cite{Cygan2015,FG98,DF99,N06}. However, there are currently very few papers that attack scheduling problems from the parameterized perspective~\cite{BodlaenderandFellows1995,FellowsandMcCartin2003,MnichandWiese2013,journals/scheduling/BevernMNW15,journals/corr/BevernNS15}. This is rather disappointing since scheduling problems seem to be particulary adequate for parameterized analyses. For one, scheduling problems which are NP-hard lack polynomial-time algorithms for finding optimal solutions, and in several applications, approximate solutions can result in big revenue losses. This gives strong motivation for computing exact solutions even if computing such solutions requires a lot of resources. Secondly, as argued above, scheduling problems typically have an abundance of natural problem parameters that can be comparatively small in practice. Thus, algorithms whose running times grow exponentially in such parameters alone should be quite useful for practical purposes.

Our aim in this paper is to help close the gap between research in parameterized complexity and the area of scheduling. We initiate a parameterized analysis on problems occurring in the setting of multi-agent scheduling~\cite{Agnetisbook}, a contemporary area which is nowadays at the cutting edge of scheduling research. We focus on the most basic case where there are only two agents, and all jobs are to be processed on a single machine. Furthermore, our parameterized analysis focuses on the scenario where the second agent has a significantly smaller number of jobs than the other. The number of jobs belonging to this agent is thus taken as a parameter, and is denoted by $k$ throughout the paper. We preform an extensive parameterized analysis for several two-agent single-machine scheduling problems with respect to this parameter, providing a clear picture of the applicability of parameterized algorithmics to these problems.

\subsection{Our contribution}
\input{contribution}

\subsection{Related work}
\input{related_work}

%% file: contribution.tex

In this work we investigate a variety of combinations for the objective functions of each agent. For each such combination, we consider the problem where each agent has a bound on his objective function, and the goal is to determine whether there exists a single-machine schedule that meets both bounds simultaneously. The objective functions we consider are (see Section~\ref{section: pre} for formal definitions):
\begin{enumerate}
\item \emph{Total weighted completion time}, where jobs have weights, and the goal is to minimize the sum of weighted completion time over the entire job set of the agent.
\item \emph{Total weighted number of tardy jobs}, where jobs have weights and due-dates, and the objective is to minimize the total weighted number of jobs that terminate after their due-date.
\item \emph{Total weighted number of just-in-time (\emph{JIT}) jobs}, where jobs have weights and due-dates, and the goal is to maximize the total weighted number of jobs that terminate precisely on their due date.
\end{enumerate}

We consider several combinations of these scheduling criteria for which the corresponding problem is NP-complete, and for each such combination we determine whether or not the corresponding scheduling problem becomes fixed-parameter tractable with respect to $k$. There are also other subtleties that we consider, such as the unit weight case or the unit processing-time case. The paper is organized according to the first agent's scheduling criteria. Thus, Section~\ref{section: TotalWeightedCompletionTimes} deals with the case where the scheduling criterion of agent 1 is the total weighted completion time, Section~\ref{section:WeightedNumberTardyJobs} focuses on problems where the first agent criterion is the total weighted number of tardy jobs, and Section~\ref{section: WeightedNumberJustinTimJobs} is concerned with problems where the first agent criterion is the total weighted number of JIT jobs. 

%% file: related_work.tex

The set of two-agent scheduling problems was first introduced by Baker and Smith~\cite{BakerandSmith2003} and Agnetis \emph{et al.}~\cite{Agnetisetal2004}. For different combinations of the scheduling criteria, Baker and Smith focus on analyzing the problem of finding a schedule that minimizes the weighted sum of the two criteria, while Agnetis \emph{et al.} focus on analyzing the problem of minimizing the first agent criterion while keeping the value of the second agent criterion not greater than a given bound. Following these two fundamental papers, numerous researchers have studied different combinations of multi-agent scheduling problems, see \emph{e.g.}~\cite{Kov2012,Leeetal2009,Leungetal2010,MorandMosheiov2011,Yuanetal,Yin2016}. Detailed surveys of these problems appear in Perez-Gonzalez and Framinan~\cite{PerezGonzalezandFraminan} and in a recent book by Agnetis \emph{et al.}~\cite{Agnetisbook}. We give further detail of the results that are more directly related to our work in the appropriate sections of the remainder of the paper.

%% file: preliminaries.tex

In this section we introduce the notation and terminology that will be used throughout the paper. In particular, we provide concrete definitions for the problems we study, as well as a very brief introduction to the theory of parameterized complexity.

\subsection{Scheduling notation and problem definitions}

In all problems considered in this paper, the input consists of two sets of jobs that have to be processed non-preemptively on a single machine. The first set $\mathcal{J}^{(1)}=\{J_{1}^{(1)},...,J_{n}^{(1)}\}$ belongs to agent 1, while the second set $\mathcal{J}^{(2)}=\{J_{1}^{(2)},...,J_{k}^{(2)}\}$ belongs to agent 2. We assume that $k \leq n$, and for practical purposes one should think of $k$ as much smaller than $n$. Let $p_{j}^{(i)}$ be a positive integer denoting the \emph{processing time} of job $J_{j}^{(i)}$. Moreover, when relevant, let $d_{j}^{(i)}$ and $w_{j}^{(i)}$ be two positive integers representing the \emph{due date} and the \emph{weight} of job $J_{j}^{(i)}$, respectively. A \emph{schedule} $\sigma$ of $\mathcal{J}^{(1)} \cup \mathcal{J}^{(2)}$ is a set of disjoint time intervals $I_j^{(i)}=(C_j^{(i)}-p_j^{(i)},C_j^{(i)}]$ for $j=1,...,n$ if $i=1$ and $j=1,...,k$ if $i=2$, where $I_j^{(i)}$ represents the time interval in $\sigma$ where job $J_j^{(i)}$ is processed on the single machine. Note that $C_{j}^{(i)}$ represents the completion time of $J_j^{(i)}$. In case, one or both of the agents jobs have due dates, we will use $L_j^{(i)}=C_j^{(i)}-d_j^{(i)}$ to denote the \emph{lateness} of job $J_j^{(i)}$, and we set $L_{\max}^{(i)} = \max_j L_j^{(i)}$. If $L_j^{(i)}\leq 0$, then job $J_j^{(i)}$ is an \emph{early} job in $\sigma$, and otherwise it is \emph{tardy}. Accordingly, the set $\mathcal{E}^{(i)}=\{J_j^{(i)}\in \mathcal{J}^{(i)}|L_j^{(i)}\leq 0\}$ is the set of early jobs in $\sigma$ that belongs to agent $i$, and the set $\mathcal{T}^{(i)}=\{J_j^{(i)}\in \mathcal{J}^{(i)}|L_j^{(i)} > 0\}$ is the set of tardy jobs that belong to agent $i$. We also use $\widehat{\mathcal{E}}^{(i)}$ to denote the set $\widehat{\mathcal{E}}^{(i)}=\{J_j^{(i)}\in \mathcal{J}^{(i)}|L_j^{(i)}=0\}$.

The quality of a schedule is measured by two different criteria, one per each agent. We focus on problems for which either one of the two agents criteria may be either one of the following three possibilities:
\begin{enumerate}
\item The \emph{weighted sum of completion times}, denoted by $\sum w_{j}^{(i)}C_{j}^{(i)}$.
\item The \emph{weighted number of tardy jobs}, where job $J_{j}^{(i)}$ is said to be \emph{tardy} if $C_{j}^{(i)}>d_{j}^{(i)}$. We use a binary indicator variable $U_{j}^{(i)}$ which indicates whether or not $J_{j}^{(i)}$ is tardy, and $\sum w_{j}^{(i)}U_{j}^{(i)}$ denotes the weighted number of tardy jobs of agent $i$.
\item The \emph{weighted number of just-in-time (JIT) jobs}, where job $J_{j}^{(i)}$ is said to be \emph{just-in-time} if $C_{j}^{(i)}=d_{j}^{(i)}$. We use a binary indicator variable $E_{j}^{(i)}$ which indicates whether or not $J_{j}^{(i)}$ is just-in-time, and $\sum w_{j}^{(i)}E_{j}^{(i)}$ denotes the weighted number of just-in-time jobs of agent $i$.
\end{enumerate}

Note that while first two criteria are minimization criteria, the latter is a maximization criterion. For each possible combination of the criteria above, we consider the decision problem where we are given two positive integer bounds $A_1$ and $A_2$, one for each agent, and we need to find if there exists a job schedule in which both bounds are met. In case the scheduling criterion is the sum of weighted completion times or weighted number of tardy jobs, the bound $A_i$ is regarded as an upper-bound, while for weighted number of just-in-time jobs it is a lower-bound. We refer to such a job schedule, if it exists, as a \emph{feasible} job schedule. Using the standard three field notation in scheduling, we denote this set of problems by~$1 \,\big\vert\, \mathbb{C}^{(1)}, \mathbb{C}^{(2)} \,\big\vert\,$, where
$$
\mathbb{C}^{(i)}\in \left\{\,\sum w_{j}^{(i)}C_{j}^{(i)} \leq A_i, \,\sum w_{j}^{(i)}U_{j}^{(i)} \leq A_i, \,\sum w_{j}^{(i)}E_{j}^{(i)} \geq A_i\right\},
$$
for $i=1,2$. We will sometimes consider special cases of these problems, and when doing so we use the middle field to denote restrictions on our input. For example, the $1 \,\big\vert\, p^{(1)}_j =1, \sum w_{j}^{(1)}C_{j}^{(1)} \leq A_1, \sum C_{j}^{(2)} \leq A_2 \,\big\vert\,$ problem is the problem where the scheduling criterion for the first agent is the weighted sum of completion times, for the second agent is the (unweighed) sum of completion time and agent 1 has jobs with unit processing times.

\subsection{Basic concepts in parameterized complexity theory}

The main objective in parameterized complexity theory is to analyze the tractability of NP-hard problems with respect to input parameters that are not necessarily related to the total input size. Thus, problem instances are not only measured in terms of their input size $n$, but also in terms of an additional parameter~$k$. In this context, a problem is said to be tractable, or \emph{fixed-parameter tractable (FPT)}, if there is an algorithm that solves each instance of size $n$ and parameter $k$ in $f(k) \cdot n^{O(1)}$ time. Here, the function $f()$ can be any arbitrary computable (\emph{e.g.} exponential) function so long as it depends only on $k$, and the exponent in $n^{O(1)}$ is independent of $k$. The reader is referred to the excellent texts on the subject for more information~\cite{Cygan2015,FG98,DF99,N06}.

In parameterized complexity, a running-time of $2^{O(k)} \cdot n^3$ is considered tractable, and even $2^{2^{2^{O(k)}}} \cdot n^{100}$ is considered tractable. Note that while the above definition might allow some quite large running-times, when $k$ is sufficiently smaller than $n$, any such run-time drastically outperforms more common algorithms with running-times of $n^{O(k)}$ or $n^{O(k^{2})}$, for example. Moreover, a running time of, say, $2^{O(k)} \cdot n^3$, with moderate constants in the exponent can be quite fast in practice. In any case, parameterized complexity provides the most convenient form of analyzing the complexity an NP-hard problem with respect to the size of a given parameter. In our context, the parameter of each instance will always be the number of jobs of agent~2. Thus, we consider the setting where agent~1 has significantly more jobs to schedule, but nevertheless we still wish to meet both agents criteria.

Note that if a problem is NP-hard already for constant values of its parameter, then a fixed-parameter tractable algorithm for the problem will imply that P=NP. Thus, in our context, if we show that one of the problems we consider is already NP-hard when agent 2 has a constant number of jobs, this excludes the possibility that the problem has a fixed-parameter tractable algorithm under the assumption of P$\neq$NP. For showing such hardness results, we will use the classical NP-complete \textsc{Partition} problem~\cite{GJ79}, often used in the context of scheduling problems:
\begin{definition}[The \textsc{Partition} problem]
\label{Definition: Partition}
Given a set $X = \{x_1,\ldots,x_m\}$ of positive integers (encoded in binary) with $\sum^m_{j=1} x_j = 2z$, determine whether $X$ can be partitioned into two sets $S_1$ and $S_2$ such that $\sum_{x_j \in S_1} x_j = \sum_{x_j \in S_2} x_j = z$.
\end{definition} 

%% file: completion.tex

In this section we study the $1 \,\big\vert\, \sum w_{j}^{(1)}C_{j}^{(1)} \leq A_1, \, \mathbb{C}^{(2)}  \,\big\vert\,$ problem where $\mathbb{C}^{(2)}$ can be any of the three scheduling criteria discussed in Section~\ref{section: pre}. We first show that all the three corresponding problems are unlikely to admit a fixed-parameter algorithm, since they are all NP-complete even for $k=1$ (\emph{i.e}, agent 2 has a single job) as we show in Theorem \ref{theorem: single job bob}. This motivates us to study four special cases: We show that in case the jobs of agent 1 all have unit weight, the problem becomes FPT when the criteria for agent 2 is either weighted sum of completion times or weighted number of tardy jobs. However, when the criteria of agent 2 is the number of just-in-time jobs, the problem remains intractable in this case as well. We also provide an FPT algorithm for the case where both criteria are the weighted sum of completion times, and agent 1 has jobs with unit processing times.

\subsection{Intractability of the general problem with respect to $k$}
\label{sec: intractability C:C}%
\input{CC}%

\subsection{An FPT algorithm for the $1 \,\big\vert\, \sum C_{j}^{(1)} \leq A_1, \, \sum w_{j}^{(2)} C_{j}^{(2)} \leq A_2 \,\big\vert\,$ problem}
\label{section: CC2}%
\input{CC2}

\subsection{An FPT algorithm for $1 \,\big\vert\, p_{j}^{(1)}=1, \sum w_{j}^{(1)}C_{j}^{(1)} \leq A_1, \, \sum w_{j}^{(2)}C_{j}^{(2)} \leq A_2 \,\big\vert\,$}
\label{section: CC3}%
\input{CC3}

\subsection{The $1 \,\big\vert\, \sum C_{j}^{(1)} \leq A_1, \, \sum w_{j}^{(2)} U_{j}^{(2)} \leq A_2 \,\big\vert\,$ problem}
\input{CU}

\subsection{Intractability of the $1 \,\big\vert\, \sum C_{j}^{(1)} \leq A_1, \, \sum E_{j}^{(2)} \geq A_2 \,\big\vert\,$ problem}
\label{section: CE1}%
\input{CE}

%% file: CC.tex

The fact that the single agent $1\,\big\vert\, \,\big\vert\, \sum w_{j}C_{j}$ problem is solvable in $O(n\log n)$ time (see Smith~\cite{Smith}) gives us some hope that at least one of the $1 \,\big\vert\, \sum w_{j}^{(1)}C_{j}^{(1)} \leq A_1, \, \mathbb{C}^{(2)} \,\big\vert\,$ problems is tractable when $k$ is small. Unfortunately, in the following theorem we show that this is not the case for all three criteria, even if the second agent has a single job of a unit weight. We will show this via a reduction from the NP-complete \textsc{Partition} problem (see Definition~\ref{Definition: Partition}).
\begin{theorem}
\label{theorem: single job bob}
The $1 \,\big\vert\, \sum w_{j}^{(1)}C_{j}^{(1)} \leq A_1, \, \mathbb{C}^{(2)} \,\big\vert\,$ problem is NP-complete for $k = 1$, when $\mathbb{C}^{(2)}$ is either $\sum C_{j}^{(2)} \leq A_2$, $\sum U_{j}^{(2)} \leq A_2$, or $\sum E_{j}^{(2)} \geq A_2$.
\end{theorem}

\begin{proof}
We provide a reduction from the NP-complete \textsc{Partition} problem defined above. Given an instance $(X,z)$ to the \textsc{Partition} problem, with $X=\{x_1,\ldots,x_m\}$, we construct the following two-agent scheduling instance: Agent 1 will have $n=m$ jobs and agent 2 will have a single job (\emph{i.e.}, $k=1$). For $j=1,...,n$, we set $p^{(1)}_j = w^{(1)}_j = x_j$. Moreover, we set $p^{(2)}_1 = 1$, and in case $\mathbb{C}^{(2)}\in\{\sum U_{j}^{(2)} \leq A_2, \sum E_{j}^{(2)} \geq A_2\}$, we also set $d^{(2)}_1 = z + 1$. The bound on the total weighted completion time of agent 1 is set to $A_1=z+ \sum^n_{i=1} \sum_{j=1}^i x_i x_j$. The bound of agent 2 depends on his scheduling criterion: If it is the sum of completion times, we set $A_2=z + 1$, if it is the weighted number of tardy jobs we set $A_2=0$, and if it is the weighted number of JIT jobs we set $A_2=1$.

Now suppose that $X$ can be partitioned into two sets $S_1$ and $S_2$ with $\sum_{x_i \in S_1} x_i = \sum_{x_i \in S_2} x_i = z$. We construct a schedule $\sigma$ where we first schedule all agent's 1 jobs corresponding to elements of $S_1$ in an arbitrary order, followed by job $J_{1}^{(2)}$, followed by all of the jobs of agents 1 corresponding to the elements of $S_2$ in an arbitrary order. Observe that job $J_{1}^{(2)}$ completes in $\sigma$ after $\sum_{x_i \in S_1} x_i + 1 = z + 1$ time units, and so the bound of agent 2 is met in all three criteria $\mathbb{C}^{(2)}$. To see that the first agent bound is met as well, observe that if we exclude $J_{1}^{(2)}$ from $\sigma$ then the total weighted completion time of agent 1 jobs is precisely $\sum^n_{i=1} \sum_{j=1}^i x_i x_j$. Adding job $J_{1}^{(2)}$ increases the completion time of each of the first agent jobs that correspond to elements of $S_2$ by a unit. Thus, $J_{1}^{(2)}$ contributes precisely $\sum_{x_i \in S_2} x_i = z$ to the total weighted completion time of agent 1 jobs, and so the first agent bound on the total weighted completion time is met as well.

For the other direction, suppose there is a feasible schedule $\sigma$ for any possible option of $\mathbb{C}^{(2)}$. Let $\mathcal{J}^{(1)}_1$ denote the set of agent 1 jobs that are scheduled before job $J_{1}^{(2)}$ in $\sigma$, and let $\mathcal{J}^{(1)}_2$ denote agent 1 remaining jobs. Since agent 2 bound is satisfied in $\sigma$, in each of the three possible criteria it must be that $\sum_{J^{(1)}_j \in \mathcal{J}^{(1)}_1} p^{(1)}_j \leq z$. Moreover, note that the total weighted completion time of agent 1 jobs is $\sum^n_{i=1} \sum_{j=1}^i x_i x_j \,+\, \sum_{J^{(1)}_j \in \mathcal{J}^{(1)}_2} p^{(1)}_j$. Since agent 1 bound is also met by $\sigma$, it must be that $\sum_{J^{(1)}_j \in \mathcal{J}^{(1)}_2} p^{(1)}_j \leq z$. Since the sum of all processing times of agent 1 jobs is $2z$, we get that $\sum_{J^{(1)}_j \in \mathcal{J}^{(1)}_1} p^{(1)}_j=\sum_{J^{(1)}_j \in \mathcal{J}^{(1)}_2} p^{(1)}_j= z$. Thus, setting $S_1 = \{x_j : J^{(1)}_j \in \mathcal{J}^{(1)}_1\}$ and $S_2 = \{x_j : J^{(1)}_j \in \mathcal{J}^{(1)}_2\}$ yields a solution to our \textsc{Partition} instance. \qed
\end{proof}


%% file: CC2.tex

In stark contrast to the result in Theorem \ref{theorem: single job bob}, we next show that, although being NP-complete (see Agnetis \emph{et al.}~\cite{Agnetisetal2004}), the $1 \,\big\vert\, \sum C_{j}^{(1)} \leq A_1, \, \sum w_{j}^{(2)} C_{j}^{(2)} \leq A_2 \,\big\vert\,$ problem is much easier to handle. We present an FPT algorithm for this problem for parameter $k$, using the powerful result of Lenstra concerning mixed integer linear programs~\cite{Len83}. To begin with, we will need the following lemma which can be easily derived by using a simple pair-wise interchange argument (see also Agnetis \emph{et al.}~\cite{Agnetisetal2004} that prove the same argument for the less general case where both agents jobs have unit weights):
\begin{lemma}
\label{lemma: SPT}
If there is a feasible solution for the $1 \,\big\vert\, \sum C_{j}^{(1)} \leq A_1, \, \sum w_{j}^{(2)} C_{j}^{(2)} \leq A_2 \,\big\vert\,$ problem, then there exists a feasible solution where the jobs of agent 1 are scheduled in a non-decreasing order of $p_{j}^{(1)}$, i.e., according to the shortest processing time (SPT) rule.
\end{lemma}

Consider now the $1 \,\big\vert\, \sum C_{j}^{(1)} \leq A_1, \, \sum w_{j}^{(2)}C_{j}^{(2)} \leq A_2 \,\big\vert\,$ problem. Due to Lemma~\ref{lemma: SPT}, we assume, without loss of generality that the jobs of agent 1 are numbered according to the SPT rule such that $p_{1}^{(1)}\leq p_{2}^{(1)}\leq...\leq p_{n}^{(1)}$. By allowing an additional multiplicative factor of $k!$ to the running time of our algorithm, we can focus on a reduced subproblem where the ordering of the second agent jobs is predefined. Given a subproblem, we renumber the second agent jobs according to this ordering. Then, to determine whether a feasible schedule is actually possible for the given subproblem, we only need to figure out if it is possible to interleave the two ordered sets of jobs together in a way that satisfies both agents bounds. Towards this aim, we formalize any given subproblem as a mixed integer linear program (MILP) where the number of integer variables is $k$. We then complete the proof by using the celebrated result of Lenstra~\cite{Len83} which states that determining whether a given MILP has a feasible solution is fixed-parameter tractable with respect to the number of integer variables.

To formulate a given subproblem as an MILP, we define an integer variable $x_j$ for each job $J^{(2)}_j$ representing the number of jobs belonging to agent 1 that are scheduled before $J^{(2)}_j$. Therefore, for each $1 \leq j \leq k$, we add the constraint that
\begin{equation}
\label{eqn: c1}
0\leq x_j\leq n.
\end{equation}
Moreover, for $1 \leq j \leq k-1$, we add the constraint that
\begin{equation}
\label{eqn: c2}
x_j \leq x_{j+1}.
\end{equation}

\begin{lemma}
\label{lemma: agent 1 bound}
The bound on the total completion time of the first agent jobs can be formulated by the following constraint:
\begin{equation}
\label{eqn: alice bound}
\left( \sum^n_{j=1} (n- j + 1) \cdot p^{(1)}_j \right) + \left( \sum^k_{j=1} (n - x_j) \cdot p_j^{(2)} \right) \leq A_1.
\end{equation}
\end{lemma}

\begin{proof}
Observe that the first term in the left-hand side of constraint~(\ref{eqn: alice bound}) is precisely the total completion time of the first agent jobs when no job of the second agent is scheduled at all. We now add the second agent jobs according to the intended meaning of the variables~$x_1,\ldots,x_k$. If there are $x_j$ jobs of agent 1 scheduled prior to $J^{(2)}_j$ in the presumed schedule, then $J^{(2)}_j$ causes an increase of $p^{(2)}_j$ to the completion time of $n-x_j$ jobs belonging to the first agent. Thus, adding all of agent 2 jobs causes an additional increase of $\sum^k_{j=1} (n - x_j) \cdot p_j^{(2)}$ to the total completion time of the first agent jobs. \qed
\end{proof}

The encoding of the second agent bound is a bit more involved. Specifically, for each $J_j^{(2)}$, we introduce a real-valued variable $y_j$ which we would like to be equal to the contribution of the first agent jobs to the completion time of $J_j^{(2)}$. Note that by our intended meaning for variable $x_j$, this is precisely $\sum_{i=1}^{x_j}p_i^{(1)}$. However, we cannot encode this directly as a linear constraint. We therefore introduce $n$ additional real-valued variables corresponding to $J^{(2)}_j$, denoted as $y_{ij}$ for $i \in \{1,\ldots,n\}$, which are ensured to be non-negative by adding the constraint that

\begin{equation}
\label{eqn: c4}
y_{ij} \geq 0
\end{equation}
for $i \in \{1,\ldots,n\}$ and $j \in \{1,\ldots,k\}$. The $y_{ij}$ variables are used to provide upper-bounds to the ``steps'' in the contribution of the first agent jobs as depicted in Fig.~\ref{pic:unitWeights}. Accordingly, we add the constraints

\begin{equation}
y_{ij} \geq (x_j - i + 1)\cdot(p_i^{(1)} - p_{i - 1}^{(1)})
\end{equation}
for each $i \in \{1,\ldots,n\}$ and $j \in \{1,\ldots,k\}$ (naturally, we set here $p^{(1)}_0 = 0$). Furthermore, we add the constraint that
\begin{equation}
\label{eqn: c6}
y_j \geq \sum_{i = 1}^{n} y_{ij}
\end{equation}
for any $j \in \{1,\ldots,k\}$ so that $y_j$ will equal its intended meaning.

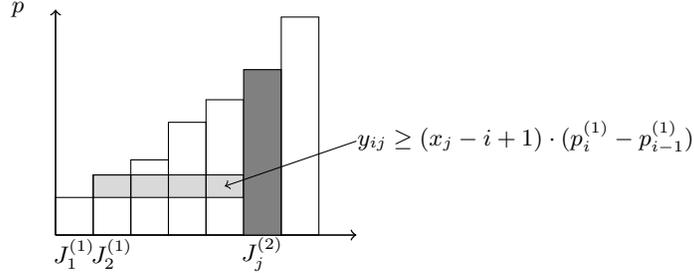
\begin{figure}
\centering
\begin{tikzpicture}
\draw[->,semithick](0,0)--(4,0);
\draw[->, semithick](0,0)--(0,3);
\draw(-.5,3)node{$p$};
\draw[fill=gray!30] (0.5,.5)rectangle(2.5,0.8);
\draw (0,0)rectangle (0.5,0.5);
\draw(.5,0)rectangle(1,0.8);
\draw (1,0)rectangle (1.5,1);
\draw (1.5,0)rectangle (2,1.5);
\draw (2,0)rectangle (2.5,1.8);
\draw[fill=gray] (2.5,0)rectangle (3,2.2);
\draw (3,0)rectangle (3.5,2.9);
\draw (0.75,-.25)node{$J_2^{(1)}$};
\draw (2.75,-.25)node{$J_{j}^{(2)}$};
\draw(0.25,-.25) node {$J_1^{(1)}$};
\draw[->](4,1.25)--(2.25,.65);
\draw (6.25,1.3) node{$y_{ij} \geq (x_j - i + 1)\cdot (p_i^{(1)} - p_{i - 1}^{(1)})  $};
\end{tikzpicture}
\caption{The contribution of $x_j$ jobs that of agent 1 which are scheduled prior to $J_j^{(2)}$ increases the completion time of this job by $\sum_{i=1}^{x_j}p_i^{(1)}$. The variable $y_{ij}$ is intended to capture the $i$'th ``step" of this contribution.}
\label{pic:unitWeights}
\end{figure}

\begin{lemma}
\label{lemma: bob bound}%
The bound on the total weighted completion time of the second agent jobs can be formulated by the following constraint:
\begin{equation}
\label{eqn: bob bound}
\left(\sum^k_{i=1}w_i^{(2)} \cdot \sum^i_{j=1} p_j^{(2)}\right)+\left(\sum^k_{j=1}w^{(2)}_j\cdot y_j\right) \,\,\leq\,\, A_2.
\end{equation}
\end{lemma}
\begin{proof}
Observe that the first term in the left-hand side of constraint~(\ref{eqn: bob bound}) is the total weighted completion time of the second agent ordered jobs $J^{(2)}_1,\ldots,J^{(2)}_k$, assuming no jobs of agent 1 are scheduled. We argue that the second term upper-bounds the contribution of agent 1 jobs. For this, it suffices to show that the contribution of agent 1 jobs to the completion time of $J^{(2)}_j$, for each $j \in \{1,\ldots,k\}$, is at most~$y_j$. We know that this contribution is $\sum_{i=1}^{x_j}p_i^{(1)}$. Since we are concerned only with feasible solutions where the constraints on variables $y_i,y_{i1}\ldots,y_{in}$ are met, we have
\begin{multline*}
y_j \,\geq\, \sum^{n}_{i=1} y_{ij} \,\geq\, \sum^{n}_{i=1} \max\{0, (x_j - i + 1)\cdot (p_i^{(1)} - p_{i- 1}^{(1)})\} \,=\, \sum^{x_j}_{i=1} (x_j - i + 1)\cdot (p_i^{(1)} - p_{i- 1}^{(1)})\\
\,\,=\,\, x_jp^{(1)}_1 + (x_j-1)(p^{(1)}_2-p^{(1)}_1) + \cdots + (p^{(1)}_{x_j}-p^{(1)}_{x_j -1}) \,\,=\,\, \sum_{i=1}^{x_j}p_i^{(1)}.
\end{multline*}
\qed
\end{proof}

To summarize, due to Lemma~\ref{lemma: SPT}, we can solve the $1 \,\big\vert\, \sum C_{j}^{(1)} \leq A_1, \, \sum w_{j}^{(2)} C_{j}^{(2)} \leq A_2 \,\big\vert\,$ problem by solving $O(k!)$ MILP formulations, each of which has only $k$ integer variables. Correctness of each of these formulations follows from Lemmas~\ref{lemma: agent 1 bound} and~\ref{lemma: bob bound}, and the analysis above. Using Lenstra's result~\cite{Len83}, we therefore obtain the following theorem:
\begin{theorem}
\label{theorem: unit weights2}
$1 \,\big\vert\, \sum C_{j}^{(1)} \leq A_1, \, \sum w_{j}^{(2)} C_{j}^{(2)} \leq A_2 \,\big\vert\,$ is fixed-parameter tractable with respect to $k$.
\end{theorem}

%% file: CC3.tex

The problem of determining whether there exists a schedule where both agents meet their respective bounds on their total weighted completion times was shown to be NP-complete even if the jobs of both agents all have unit processing times~\cite{Shabtay}. Here we complement this result, as well as the result given in Theorem~\ref{theorem: single job bob}, by showing that the problem is FPT with respect to $k$ when the jobs of the first agent have unit processing times. The following lemma is crucial for the construction of the FPT algorithm, and can be proven by a simple pairwise interchange argument.
\begin{lemma}
\label{lemma: sorting}
If there is a feasible schedule for the $1 \,\big\vert\, p_{j}^{(1)}=1, \sum w_{j}^{(1)}C_{j}^{(1)} \leq A_1, \, \sum w_{j}^{(2)}C_{j}^{(2)} \leq A_2 \,\big\vert\,$ problem, then there is a feasible schedule for the problem where the first agent jobs are scheduled in a non-increasing weight order.
\end{lemma}

Assume, without loss of generality, that $w^{(1)}_1 \geq \cdots \geq w^{(1)}_n$.  Accordingly, based on Lemma \ref{lemma: sorting}, we can restrict our search for a feasible schedule to those schedules in which job $J^{(1)}_j$ is scheduled before $J^{(1)}_{j+1}$ for all $j \in \{1,\ldots,n-1\}$. As in the
proof of Theorem~\ref{theorem: unit weights2}, by allowing an additional multiplicative factor of $k!$ to the running time of our algorithm, we can focus on a reduced subproblem where the ordering of the second agent jobs is also predefined. Given such an ordering, we renumber the second agent jobs according to the ordering.

In what follows we prove that each subproblem is FPT with respect to $k$. The proof uses the same ideas as the proof of Theorem~\ref{theorem: unit weights2} and thus is briefly presented.  Here as well we define variables $x_1,\ldots,x_k$, but this time $x_j$ represents the number of the first agent jobs that are scheduled \textit{after} $J^{(2)}_j$. Accordingly, we have to include the constraint in (\ref{eqn: c1}) and for $1 \leq j \leq k-1$ the following constraint as well:
\begin{equation}
\label{eqn: c7}
x_j \geq x_{j+1}.
\end{equation}

The bound on the weighted sum of completion times of the second agent jobs can be expressed by
\begin{equation}
\label{eqn: c7}
\left(\sum^k_{i=1} w_i^{(2)} \sum^i_{j=1} p^{(2)}_j\right)  + \left(\sum^k_{j=1}w_j^{(2)}(n-x_j)\right) \leq A_2,
\end{equation}
where the first term in the left-hand side is the weighted sum of completion times of the second agent jobs if they are scheduled one after the other at the beginning of the schedule, and the second term in the left-hand side corresponds to the contribution of the first agent jobs to the weighted sum of completion times of the second agent jobs.

To bound the total weighted completion time of the first agent jobs, we need again to introduce a set of real-valued variables $y_j,y_{j1},\ldots,y_{jn}$ for each $j \in \{1\ldots,k\}$. Here, variable $y_j$ is meant to encode the contribution of $J^{(2)}_j$ to the total weighted completion time of the first agent jobs. Note that this is precisely $p^{(2)}_j\sum^n_{i=n-x_j+1}w^{(1)}_i$. We use the variables $y_{ij}$ to encode lower-bounds on the steps of the sum $\sum^n_{i=n-x_j+1}w^{(1)}_i$, as done in the proof of Theorem~\ref{theorem: unit weights2}. Accordingly, for $j \in \{1,\ldots,k\}$, we include the set of constraints in (\ref{eqn: c6}). Moreover, for $i \in \{1,\ldots,n\}$ and $j \in \{1,\ldots,k\}$, we include the set of constraints in (\ref{eqn: c4}) and also the following set of constraints
\begin{equation}
\label{eqn: c8}
y_{ij}\geq (x_j-n+i)(w^{(1)}_{i} - w^{(1)}_{i+1}),
\end{equation}
where $w_{n+1}^{(1)}=0$ by definition. Finally, we encode the bound on the total weighted completion time of the first agent jobs by
\begin{equation}
\label{eqn: c8}
\left(\sum^{n}_{j=1} jw_j^{(1)}\right) + \left(\sum^k_{j=1} p^{(2)}_j y_j\right) \leq A_1,
\end{equation}
where the first term in the left-hand side is the weighted sum of completion time of the first agent jobs if they are scheduled one after the other at the beginning of the schedule, and the second term in the left-hand side corresponds to the contribution of the second agent jobs to the weighted sum of completion times of the first agent jobs.


\begin{theorem}
\label{theorem: unit processing}
The $1 \,\big\vert\, p_{j}^{(1)}=1, \sum w_{j}^{(1)}C_{j}^{(1)} \leq A_1, \, \sum w_{j}^{(2)}C_{j}^{(2)} \leq A_2 \,\big\vert\,$ problem is fixed-parameter tractable with respect to $k$.
\end{theorem}

%% file: CU.tex

Ng  \emph{et al.}~\cite{Ng06} and  Leung  \emph{et al.}~\cite{Leungetal2010} proved that the $1 \,\big\vert\, \sum C_{j}^{(1)} \leq A_1, \, \sum U_{j}^{(2)} \leq A_2 \,\big\vert\,$ problem is NP-complete. We next prove that the more general $1 \,\big\vert\, \sum C_{j}^{(1)} \leq A_1, \, \sum w_{j}^{(2)}U_{j}^{(2)} \leq A_2 \,\big\vert\,$ problem is FPT with respect to $k$. Our proof depends on the following easy-to-prove lemma (recall the definitions of $\mathcal{E}^{(i)}$ and $\mathcal{T}^{(i)}$ in Section~\ref{section: pre}):
\begin{lemma}
\label{structure}
If there is a feasible solution for an instance of the $1 \,\big\vert\, \sum C_{j}^{(1)} \leq A_1, \, \sum w_{j}^{(2)}  U_{j}^{(2)} \leq A_2 \,\big\vert\,$ problem, then for the same instance there exists a feasible solution in which $(i)$ the jobs of agent 1 are scheduled in a non-decreasing order of $p_{j}^{(1)}$ (i.e., according to the SPT rule); $(ii)$ the jobs in $\mathcal{E}^{(2)}$ are scheduled in a non-decreasing order of $d_{j}^{(2)}$ (i.e., according to the EDD rule); and $(iii)$ the jobs in $\mathcal{T}^{(2)}$ are scheduled last in an arbitrary order.
\end{lemma}

Consider now the $1 \,\big\vert\, \sum C_{j}^{(1)} \leq A_1, \, \sum w_{j}^{(2)}  U_{j}^{(2)} \leq A_2 \,\big\vert\,$ problem and define a set of $2^k$ subproblems corresponding to the $O(2^k)$ possible ways to partition set $\mathcal{J}^{(2)}$ into $\mathcal{E}^{(2)}$ and $\mathcal{T}^{(2)}$ such that the condition $\sum_{J_j^{(2)}\in \mathcal{T}^{(2)}}w_{j}^{(2)}\leq A_2$ holds. Due to Lemma~\ref{structure}, each subproblem reduces to an instance of the $1 \,\big\vert\, \sum C_{j}^{(1)} \leq A_1, \, L^{(2)}_{\max} \leq 0  \,\big\vert\,$ problem which includes the $n$ jobs of agent 1 and only the $O(k)$ early jobs of agent 2. In the reduced subproblem, we need to find if it is possible to schedule the jobs in $\mathcal{J}^{(1)}\cup \mathcal{E}^{(2)}$ such that all jobs in $\mathcal{E}^{(2)}$ are indeed early (\emph{i.e.}, completed not later than its due date), and $\sum C_{j}^{(1)}\leq A_1$. Thus, the fact that the $1 \,\big\vert\, \sum C_{j}^{(1)} \leq A_1, \, L^{(2)}_{\max} \leq 0  \,\big\vert\,$ problem is solvable in $O(n+k)=O(n)$ time (see Yuan \textit{et al.} ~\cite{Yuanetal}) leads to the following theorem:

\begin{theorem}
\label{theorem: single job bob3}
The $1 \,\big\vert\, \sum C_{j}^{(1)} \leq A_1, \, \sum w_{j}^{(2)} U_{j}^{(2)} \leq A_2 \,\big\vert\,$ problem is solvable in $O(2^kn)$ time.
\end{theorem}

%% file: CE.tex

Consider an instance of the $1 \,\big\vert\, \mathbb{C}^{(1)}, \, \sum E_{j}^{(2)} \geq A_2 \,\big\vert\,$ problem, with $k=1$ and $A_2=1$. If there is a feasible solution for such an instance, then the single job of agent 2 is scheduled in a JIT mode, i.e., during time interval $(d_{1}^{(2)}-p_{1}^{(2)},d_{1}^{(2)}]$. Thus, such an instance is equivalent to an instance of a $1 \,\big\vert\, \mathbb{C}^{(1)} \, \big\vert\,$ problem with a single non-availability interval (more commonly denoted by $1 \,\big\vert\, \mathbb{C}^{(1)},n-a\, \big\vert\,$). This problem is known to be NP-complete when $\mathbb{C}^{(1)}=\sum C_{j}^{(1)} \leq A_1$ (see Adiri \emph{et al.}~\cite{ACTAI::AdiriBFK1989} and Lee and Liman~\cite{ACTAI::LeeL1992}). Thus, we have the following corollary:

\begin{corollary}
\label{theorem: hardness1}
The $1 \,\big\vert\, \sum C_{j}^{(1)} \leq A_1, \, \sum E_{j}^{(2)} \geq A_2 \,\big\vert\,$ problem is \textnormal{NP}-complete even for $k=1$.
\end{corollary}

%% file: tardy.tex

In this section we study variants of our problem of the form $1 \,\big\vert\, \sum w_{j}^{(1)}U_{j}^{(1)} \leq A_1, \, \mathbb{C}^{(2)}  \,\big\vert\,$. That is, variants where the scheduling criteria of agent 1 is the weighted number of tardy jobs. Note that already the single agent $1\,\big\vert\, \sum w_{j}U_{j} \leq A \,\big\vert\,$ problem is NP-complete, even when all due dates are equal (a resulting dating back to Karp's seminal NP-completeness paper~\cite{Kar72}). Therefore, any variant of our problem when the jobs of the first agent have weights is hard.
\begin{corollary}
\label{thm: SigmaWUhardness}
The $1 \,\big\vert\, d_j=d, \, \sum w_{j}^{(1)}U_{j}^{(1)} \leq A_1, \, \mathbb{C}^{(2)}  \,\big\vert\,$ problem is \textnormal{NP}-complete for $k \geq 0$.
\end{corollary}

Due to Theorem~\ref{thm: SigmaWUhardness}, we restrict our analysis below to the unweighed $1 \,\big\vert\, \sum U_{j}^{(1)} \leq A_1, \, \mathbb{C}^{(2)}  \,\big\vert\,$ problem. We provide a fixed-parameter algorithm for the case where the criteria of agent 2 is also the (weighted) number of tardy jobs. This algorithm is then extended to the case where the jobs of agent 1 may have arbitrary weights, but both agents jobs have unit processing times. On the contrary, when the criteria for agent 2 is the weighted number of JIT jobs, we show that the problem is intractable already for highly restrictive special cases. We do not know whether the problem is fixed-parameter tractable when the criteria for agent~2 is the total weighted completion time; in the case, we can only show an $n^{O(k)}$-time algorithm. 

\subsection{An $n^{O(k)}$-time algorithm for the $1 \,\big\vert\, \sum U_{j}^{(1)} \leq A_1, \, \sum w_{j}^{(2)} C_{j}^{(2)} \leq A_2 \,\big\vert\,$ problem}
\input{UC}
\subsection{An FPT algorithm for $1 \,\big\vert\, \sum U_{j}^{(1)} \leq A_1, \, \sum w_{j}^{(2)} U_{j}^{(2)} \leq A_2 \,\big\vert\,$}
\input{UU}
\subsection{An FPT algorithm for $1 \,\big\vert\, p_{j}^{(1)}=p_{j}^{(2)}=1, \sum w_{j}^{(1)}U_{j}^{(1)} \leq A_1, \, \sum w_{j}^{(2)}U_{j}^{(2)} \leq A_2 \,\big\vert\,$}
\input{UU2}
\subsection{Intractability of the $1 \,\big\vert\, \sum U_{j}^{(1)} \leq A_1, \, \sum E_{j}^{(2)} \geq A_2 \,\big\vert\,$ problem}
\input{UE}

%% file: UC.tex

The question whether the $1 \,\big\vert\, \sum U_{j}^{(1)} \leq A_1, \, \ \sum w_{j}^{(2)}C_{j}^{(2)} \leq A_2 \,\big\vert\,$ problem is fixed-parameter tractable or not remains an open question. Nevertheless, we show below that the problem can be solved in much slower but still non-trivial $n^{O(k)}$ time. This leads to the conclusion that the problem belongs to the parameterized class XP (see \emph{e.g.}~\cite{DF99} for a formal definition), and is solvable in polynomial time when $k$ is upper bounded by a constant. We begin with the following easy lemma.

\begin{lemma}
\label{structure3}
If there is a feasible solution for an instance of the $1 \,\big\vert\, \sum U_{j}^{(1)} \leq A_1, \, \ \sum w_{j}^{(2)}C_{j}^{(2)} \leq A_2 \,\big\vert\,$ problem, then for the same instance there exists a feasible solution in which (i) $\sum U_{j}^{(1)}=A_1$;  (ii) the jobs in $\mathcal{E}^{(1)}\cup \mathcal{J}^{(2)}$ are scheduled first followed by the jobs in $\mathcal{T}^{(1)}$ that are scheduled last in an arbitrary order; and (iii) the jobs in $\mathcal{E}^{(1)}$ are scheduled according to the EDD rule (i.e., in non-decreasing order of $d_{j}^{(1)}$).
\end{lemma}

Following Lemma~\ref{structure3}, we renumber the jobs in $\mathcal{J}^{(1)}$ according to the EDD rule. Furthermore, we divide the original problem into $k!$ instances, each of which represent a different processing order of the jobs in $\mathcal{J}^{(2)}$. Consider a given instance, and assume that the jobs in $\mathcal{J}^{(2)}$ are numbered according to their processing order. For each such instance, we consider all possible partitions of $\mathcal{J}^{(1)}$ into $k+1$ subsets, $\mathcal{J}^{(1)} = \mathcal{J}_0^{(1)} \cup \cdots \cup \mathcal{J}_k^{(1)}$, and all possible sets of $k+1$ integers $\{e_0,\ldots,e_k\}$ with $e_i \leq |\mathcal{J}_i^{(1)}|$ for each $i \in \{0,\ldots,k+1\}$ and $\sum_{i=0}^{k} e_i=n-A_1$.

Note that there are $O(n^{2k+2})$ such pairs $(\{\mathcal{J}_0^{(1)},\ldots,\mathcal{J}_k^{(1)}\}, \{e_0,\ldots,e_k\})$. Given such a pair, we are looking for a restricted schedule that satisfies the following three conditions: ($i$) there are at most $e_i$ early jobs among the jobs in $\mathcal{J}_i^{(1)}$ for each $i=0,...,k$; ($ii$) job $J^{2}_{i+1}$ is scheduled right after the $e_i$ early jobs within $\mathcal{J}_i^{(1)}$ for $i=0,...,k-1$; and $(iii)$ $\sum w_{j}^{(2)}C_{j}^{(2)} \leq A_2$. Note that the instance corresponding to the particular processing order of $\mathcal{J}^{(2)}$ has a feasible schedule iff such a restricted schedule corresponding to some pair $(\{\mathcal{J}_0^{(1)},\ldots,\mathcal{J}_k^{(1)}\}, \{e_0,\ldots,e_k\})$ exists. Below we show how to compute a restricted schedule, if it exists, in $O(n\log n)$ time. This will yield the following theorem:

\begin{theorem}
\label{theorem: ucuc}
The $1 \,\big\vert\, \sum U_{j}^{(1)} \leq A_1, \, \sum w_{j}^{(2)} C_{j}^{(2)} \leq A_2 \,\big\vert\,$ problem is solvable in $O(k!n^{2k+1}\log n)$ time.
\end{theorem}

Consider some pair $(\{\mathcal{J}_0^{(1)},\ldots,\mathcal{J}_k^{(1)}\}, \{e_0,\ldots,e_k\})$ as above. Note that if our goal was only to find a schedule for there are at most $e_i$ early jobs in $\mathcal{J}_i^{(1)}$ for each $i$, then this translates to solving $k$ disjoint instances of the single agent $1 \,\big\vert\, \sum U_{j} \leq n_i-e_i \,\big\vert\,$ problem over each set of jobs $\mathcal{J}_i^{(1)}$, where $n_i = |\mathcal{J}_i^{(1)}|$. Each such instance can be solved in $O(n_i \lg n_i)$ time by a slight modification of the classical algorithm of Moore~\cite{Moore}, which gives us a total of $O(n \log n)$ time for all $k+1$ instances. Furthermore, Moore's algorithm computes the schedule with minimum makespan (\emph{i.e.}, final completion time) amongst all schedules with at most $e_i$ tardy jobs. Thus, composing the $k+1$ schedules into a single schedule for both agents by scheduling $J_{i+1}^{(2)}$ after the final job scheduled in $\mathcal{J}_i^{(1)}$, gives us a schedule which minimizes $\sum w_{j}^{(2)}C_{j}^{(2)}$ over all schedules which satisfy properties $(i)$ and $(ii)$ above. Thus, in $O(n \lg n)$ time we can determine whether there exists a restricted schedule corresponding to $(\{\mathcal{J}_0^{(1)},\ldots,\mathcal{J}_k^{(1)}\}, \{e_0,\ldots,e_k\})$, and so Theorem~\ref{theorem: ucuc} holds.

We mention that this algorithm can slightly be improved if the jobs of agent 2 have unit weights. In this case, we know that it is optimal to order these jobs  in a non-decreasing order of $p_j^{(2)}$, \emph{i.e.}, according to the shortest processing time (SPT) rule. Thus, we do not have to try out all possible orderings of $\mathcal{J}^{(2)}$, reducing the time complexity of the algorithm above by a factor of $O(k!)$.

\begin{corollary}
\label{corollary: ucuc}
The $1 \,\big\vert\, \sum U_{j}^{(1)} \leq A_1, \, \sum C_{j}^{(2)} \leq A_2 \,\big\vert\,$ problem is solvable in $O(n^{2k+1}\log n)$ time.
\end{corollary} 

%% file: UU.tex

We next show that the $1 \,\big\vert\, \sum U_{j}^{(1)} \leq A_1, \, \sum w_{j}^{(2)} U_{j}^{(2)} \leq A_2 \,\big\vert\,$ problem is FPT with respect to $k$. Our proof depends on the following easy-to-prove lemma:
\begin{lemma}
\label{structure2}
If there is a feasible solution for an instance of the $1 \,\big\vert\, \sum U_{j}^{(1)} \leq A_1, \, \sum w_{j}^{(2)} U_{j}^{(2)} \leq A_2 \,\big\vert\,$ problem, then for the same instance there exists a feasible solution in which the jobs in $\mathcal{E}^{(1)}\cup \mathcal{E}^{(2)}$ are scheduled first according to the EDD rule (i.e., in non-decreasing order of $d_{j}^{(i)}$), followed by the jobs in $\mathcal{T}^{(1)}\cup \mathcal{T}^{(2)}$ that are scheduled last in an arbitrary order.
\end{lemma}

Consider now an instance of the $1 \,\big\vert\, \sum U_{j}^{(1)} \leq A_1, \, \sum w_{j}^{(2)} U_{j}^{(2)} \leq A_2 \,\big\vert\,$ problem, and define a set of $2^k$ instances corresponding to the $O(2^k)$ possible ways to partition set $\mathcal{J}^{(2)}$ into $\mathcal{E}^{(2)}$ and $\mathcal{T}^{(2)}$ such that the feasibility condition $\sum_{J_j^{(2)}\in \mathcal{T}^{(2)}}w_{j}^{(2)}\leq A_2$ holds in each instance. Due to Lemma~\ref{structure2}, each of these instances is an instance of the $1\left\vert \sum U_{j}^{(1)}\leq A_1,L_{\max }^{(2)}\leq 0\right\vert$ problem in which we need to find a feasible schedule for agent 1 subject to scheduling the set of $O(k)$ jobs in $\mathcal{E}^{(2)}$ such that they are all early. Agnetis \emph{et al.} showed that the $1\left\vert \sum U_{j}^{(1)}\leq A_{1},L_{\max }^{(2)}\leq 0\right\vert$ problem is solvable in $O(n \log n +k\log k)= O(n \log n)$ time~\cite{Agnetisetal2004}. Thus, we obtain the following:

\begin{theorem}
\label{theorem: uuu}
The $1 \,\big\vert\, \sum U_{j}^{(1)} \leq A_1, \, \sum w_{j}^{(2)} U_{j}^{(2)} \leq A_2 \,\big\vert\,$ problem is solvable in $O(2^kn \log n)$ time.
\end{theorem} 

%% file: UU2.tex

The $1 \,\big\vert\, \sum w_{j}^{(1)}U_{j}^{(1)} \leq A_1, \, \sum w_{j}^{(2)} U_{j}^{(2)} \leq A_2 \,\big\vert\,$ problem is NP-complete even for the case of unit processing time~\cite{Shabtay}. Next we complement this result, as well as Theorem~\ref{thm: SigmaWUhardness}, by showing that this problem is FPT with respect to $k$. First observe that Lemma~\ref{structure2} holds here as well. Thus, we again create $2^{k}$ instances from our $1 \,\big\vert\, p_{j}^{(1)}=p_{j}^{(2)}=1, \sum w_{j}^{(1)}U_{j}^{(1)} \leq A_1, \, \sum w_{j}^{(2)}U_{j}^{(2)} \leq A_2 \,\big\vert\,$ instance, where in each instance the set $\mathcal{E}^{(2)} \subseteq \mathcal{J}^{(2)}$ may vary. In any given instance of these $2^k$ instances we may assume that $\sum_{J_j^{(2)}\in \mathcal{T}^{(2)}}w_{j}^{(2)}\leq A_2$. Furthermore, again due to Lemma~\ref{structure2}, each of these instances is in fact an instance of the $1 \,\big\vert\, p_{j}^{(1)}=p_{j}^{(2)}=1, \sum w_{j}^{(1)}U_{j}^{(1)} \leq A_1, \, L_{\max }^{(2)}\leq 0\,\big\vert\,$ problem in which we need to find a feasible schedule for agent 1 subject to scheduling the set of $O(k)$ jobs in $\mathcal{E}^{(2)}$ such that they are all early. This latter  problem is solvable in $O(\max\{n \log n, k\})=O(n \log n)$ time~\cite{Shabtay}. Thus, we get:
\begin{theorem}
\label{theorem: uuu2}
The $1 \,\big\vert\, p_{j}^{(1)}=p_{j}^{(2)}=1, \sum w_{j}^{(1)}U_{j}^{(1)} \leq A_1, \, \sum w_{j}^{(2)} U_{j}^{(2)} \leq A_2 \,\big\vert\,$ problem is solvable in $O(2^kn \log n)$ time.
\end{theorem} 

%% file: UE.tex
Following the observation made in Section~\ref{section: CE1}, we can conclude that an instance of the $1 \,\big\vert\, d_j^{(1)}=d, \, \sum U_{j}^{(1)} \leq A_1, \, \ \sum E_{j}^{(2)} \geq A_2 \,\big\vert\,$ problem with $k=1$ and $A_2=1$ is equivalent to an instance of the $1 \,\big\vert\, d_j^{(1)}=d, \, \sum U_{j}^{(1)} \leq A_1, n-a \,\big\vert\,$ problem, which is known to be NP-complete (see Lee~\cite{journals/jgo/Lee96}). Thus, we have the following corollary:

\begin{corollary}
\label{theorem: hardness2}
The $1 \,\big\vert\, d_j^{(1)}=d, \, \sum U_{j}^{(1)} \leq A_1, \, \ \sum E_{j}^{(2)} \geq A_2 \,\big\vert\,$ problem is \textnormal{NP}-complete for $k\geq1$.
\end{corollary}



%% file: jit.tex

We next consider problems of the form $1 \,\big\vert\, \sum w_{j}^{(1)}E_{j}^{(1)} \geq A_1, \, \mathbb{C}^{(2)}  \,\big\vert\,$, \emph{i.e.}, problems where the criteria of agent 1 is the total weighted number of JIT jobs. Recall that a job $J_j^{(i)}$ is scheduled in JIT mode if it is scheduled precisely at the time interval $(d_j^{(i)}-p_j^{(i)},d_j^{(i)}]$. We will show that when either $\mathbb{C}^{(2)} = \sum w_{j}^{(2)} U_{j}^{(2)} \leq A_2$, or $\mathbb{C}^{(2)}= \sum w_{j}^{(2)} E_{j}^{(2)} \geq A_2$, the problem is fixed-parameter tractable in $k$. We also show that when $\mathbb{C}^{(2)} = \sum w_{j}^{(2)} C_{j}^{(2)} \leq A_2$ the problem is fixed-parameter tractable if the first agent jobs are unweighed, while the more general case (where the first agent jobs are weighted) is left open.


\subsection{The $1 \,\big\vert\, \sum E_{j}^{(1)} \geq A_1, \, \sum w_{j}^{(2)}C_{j}^{(2)} \leq A_2 \,\big\vert\,$ problem}
\label{section: EC}
\input{EC}
\subsection{The $1 \,\big\vert\, \sum w_{j}^{(1)}E_{j}^{(1)} \geq A_1, \, \sum w_{j}^{(2)}U_{j}^{(2)} \leq A_2 \,\big\vert\,$ problem}
\input{EU}
\subsection{The $1 \,\big\vert\, \sum w_{j}^{(1)}E_{j}^{(1)} \geq A_1, \, \sum w_{j}^{(2)}E_{j}^{(2)} \geq A_2 \,\big\vert\,$ problem}
\input{EE}

%% file: EC.tex

The $1 \,\big\vert\, \sum E_{j}^{(1)} \geq A_1, \, \sum C_{j}^{(2)} \leq A_2 \,\big\vert\,$ problem is known to be NP-complete~\cite{Sha}. We next prove that the more general $1 \,\big\vert\, \sum E_{j}^{(1)} \geq A_1, \, \sum w_{j}^{(2)}C_{j}^{(2)} \leq A_2 \,\big\vert\,$ problem is FPT with respect to $k$. As usual, we begin with an easy-to-prove lemma:
\begin{lemma}
\label{lemma: auxilary_jit1}
In any feasible schedule for the $1 \,\big\vert\, \sum E_{j}^{(1)} \geq A_1, \, \sum w_{j}^{(2)}C_{j}^{(2)} \leq A_2 \,\big\vert\,$ problem the jobs in $\mathcal{E}^{(1)}$ are scheduled in non-decreasing due date (EDD) order. Moreover, if there is a feasible solution for the problem, then there is a feasible solution in which $|\widehat{\mathcal{E}}^{(1)}|=\sum E_{j}^{(1)}=A_1$, and the jobs in set $\widehat{\mathcal{T}}^{(1)}=\mathcal{J}^{(1)}\setminus \widehat{\mathcal{E}}^{(1)}$ are scheduled last in an arbitrary order.
\end{lemma}

Following Lemma~\ref{lemma: auxilary_jit1}, we assume without loss of generality that the jobs in $\mathcal{J}^{(1)}$ are indexed according to the EDD rule, \emph{i.e.}, $d_1^{(1)}\leq d_2^{(1)}\leq ... \leq d_n^{(1)}$. We will show that for any fixed ordering of the jobs of agents 2, we can determine in polynomial time whether there exists a feasible schedule where the relative order of agent 2 jobs is exactly this ordering. Since there are $k!$ orderings of the jobs of agent 2, this will imply that the $1 \,\big\vert\, \sum E_{j}^{(1)} \geq A_1, \, \sum w_{j}^{(2)}C_{j}^{(2)} \leq A_2 \,\big\vert\,$ problem is FPT with respect to $k$.

Consider any fixed ordering $J^{(2)}_1,\ldots,J^{(2)}_k$ of the jobs of agent 2. For convenience purposes, let $J_{0}^{(1)}$ and $J_{n+1}^{(1)}$ be two dummy jobs of agent 1 with $p_{0}^{(1)}=w_{0}^{(1)}=d_{0}^{(1)}=p_{n+1}^{(1)}=w_{n+1}^{(1)}=0$ and $d_{n+1}^{(1)}=\infty$. Consider any given feasible schedule $\sigma$ with $J_a^{(1)}$ and $J_b^{(1)}$ being two jobs of agent 1 that are scheduled in JIT mode, with no other jobs of agent 1 scheduled in between them. Obviously, $d_b^{(1)}-d_a^{(1)}-p_b^{(1)}\geq 0$, as otherwise $J_a^{(1)}$ and $J_b^{(1)}$ cannot be both scheduled in JIT mode. For $\ell \in \{0,\ldots,k-1\}$, let $k(a,b,\ell)$ denote the number of jobs in a maximal consecutive subsequence of jobs in $\{J^{(2)}_{\ell+1},\ldots,J^{(2)}_k\}$ with total processing time not greater than $d_b^{(1)}-d_a^{(1)}-p_b^{(1)}$. Then the following lemma holds:
\begin{lemma}
\label{lemma: auxilary_jit2}
Suppose there is a feasible schedule in which $J_a^{(1)}$ and $J_b^{(1)}$ are scheduled in JIT mode with no other jobs of agent 1 scheduled between them, and there are exactly $\ell$ jobs belonging to agent 2 that are scheduled prior to $J_a^{(1)}$. Then there exists a feasible schedule where the sequence of jobs $J^{(2)}_{\ell+1},\ldots,J^{(2)}_{\ell+k(a,b,\ell)}$ is scheduled right after the completion time of job $J_a^{(1)}$, and no other jobs are scheduled prior to the completion of job $J_b^{(1)}$.
\end{lemma}

For a pair of jobs $J^{(1)}_a,J^{(1)}_b \in \mathcal{J}^{(1)}$ and $\ell \in \{0,\ldots,k-1\}$, let $w(a,b,\ell)$ denote the minimum contribution of the job sequence $J^{(2)}_{\ell+1},\ldots,J^{(2)}_{\ell+k(a,b,\ell)}$ to the total weighted completion time of agent 2 in any feasible schedule $\sigma$. If no such schedule exists, define $w(a,b,\ell) = \infty$. According to Lemma~\ref{lemma: auxilary_jit2} we can compute $w(a,b,\ell)$  by using the following formula:
\begin{equation}
\label{eqn: c11}
w(a,b,\ell) = \sum_{j=\ell+1}^{\ell+k(a,b,\ell)} w_{j}^{(2)}C_{j}^{(2)} = \sum_{j=\ell+1}^{\ell+k(a,b,\ell)} w_{j}^{(2)}\left(d_a^{(1)}+\sum_{i=\ell+1}^{j}p_{i}^{(2)}\right).
\end{equation}

Now, for $b, e \in \{0,\ldots,n+1\}$, and $\ell \in \{0,\ldots,k-1\}$, let $W(b,e,\ell)$ denote the minimum total weighted completion time of the first $\ell$ jobs of agent 2, among all partial schedules on the job set $\{J_1^{(1)},\ldots,J_b^{(1)},J_1^{(2)},\ldots,J_{\ell}^{(2)}\}$, such that $e=|\widehat{\mathcal{E}}^{(1)}|=\sum E_j^{(1)}$ and $J_{b}^{(1)}\in \widehat{\mathcal{E}}^{(1)}$ is the last job to be scheduled. (Again, set $W(b,e,\ell) = \infty$ if no such schedule exists.) Note that by definition, the completion time of such a feasible partial schedule is exactly at time $d_b^{(1)}$. The value $W(b,e,\ell)$ can be computed using the following recursion that considers all possible ways of appending partial schedules that have fewer jobs of agent 2 and one less JIT job of agent 1:
\begin{equation}
\label{eqn: 15}
W(b,e,\ell)=\min_{\substack{0\leq a \leq b-1 \\
0 \leq \ell' \leq \ell}}
\begin{cases}
W(a,e-1,\ell')+w(a,b,\ell') &  : \, d_b^{(1)}-d_a^{(1)}-p_b^{(1)} \geq 0, \, \ell=\ell'+ k(a,b,\ell').\\
\infty & : \, \text{otherwise}.
\end{cases}
\end{equation}

Our algorithm computes all possible $W(b,e,\ell)$ values for $b\in \{0,...,n+1\}$, $e\in \{0,...,A_1\}$, and $\ell\in \{0,...,k\}$, using the recursion given in Equation~\ref{eqn: 15}. For this, it computes in a preprocessing step all values $w(a,b,\ell)$ and $k(a,b,\ell)$, for $0 \leq a < b \leq n+1$ and $0 \leq \ell \leq k$. The base cases of the recursion are given by $W(0,0,0)=0$, and $W(0,e,\ell)=\infty$ for $e\neq 0$ or $\ell \neq 0$. We report that there exists a feasible schedule for our $1 \,\big\vert\, \sum E_{j}^{(1)} \geq A_1, \, \sum C_{j}^{(2)} \leq A_2 \,\big\vert\,$ instance, restricted to the relative fixed ordering $J^{(2)}_1,\ldots,J^{(2)}_k$ for the jobs of agent 2, iff $W(n+1,A_1,k) \leq A_2$.

Correctness of our algorithm is immediate from the above discussion. Let us now analyze its time complexity.  There are $k!$ ways to order the jobs of agent 2. For each such order, we compute all values $W(b,e,\ell)$, $w(a,b,\ell)$ and $k(a,b,\ell)$. All values $k(a,b,\ell)$ can be computed straightforwardly in $O(n^2k^2)$ time. Using these values, and Equation~\ref{eqn: c11}, all values $w(a,b,\ell)$ can also be computed in $O(n^2k^2)$ time. Finally, using dynamic programming along with Equation~\ref{eqn: 15}, computing all values $W(b,e,\ell)$ requires $O(n^3k^2)$ time. Thus, for each fixed ordering of $\mathcal{J}^{(2)}$, we spend a total of $O(n^3k^2)$ time. All together, this gives us an $O(k!k^2n^3)$ time algorithm.

\begin{theorem}
\label{theorem: EC1}
The $1 \,\big\vert\, \sum E_{j}^{(1)} \geq A_1, \, \sum w_{j}^{(2)}C_{j}^{(2)} \leq A_2 \,\big\vert\,$ problem is solvable in $O(k!k^2n^3)$ time.
\end{theorem} 

%% file: EU.tex

We next show how to modify the ideas used in Section~\ref{section: EC} so that they apply to the $1 \,\big\vert\, \sum w_{j}^{(1)}E_{j}^{(1)} \geq A_1, \, \sum w_{j}^{(2)}U_{j}^{(2)} \leq A_2 \,\big\vert\,$ problem. We begin with the following analog of Lemma~\ref{lemma: auxilary_jit1}:
\begin{lemma}
\label{lemma: auxilary_jit3}
In any feasible schedule for an instance of the $1 \,\big\vert\, \sum w_{j}^{(1)}E_{j}^{(1)} \geq A_1, \, \sum w_{j}^{(2)}U_{j}^{(2)} \leq A_2 \,\big\vert\,$ problem, the jobs in $\widehat{\mathcal{E}}^{(1)} \cup \mathcal{E}^{(2)}$ are scheduled in an earliest due date (EDD) order. Moreover, if there is a feasible solution for the instance, then there is a feasible solution in which the jobs in set $\widehat{\mathcal{T}}^{(1)}\cup \mathcal{T}^{(2)}$  are scheduled last in an arbitrary order.
\end{lemma}

Following Lemma~\ref{lemma: auxilary_jit3}, we assume that the jobs in $\mathcal{J}^{(1)}$ are numbered according to the EDD rule such that $d_1^{(1)}\leq d_2^{(1)}\leq ... \leq d_n^{(1)}$. We create $O(2^k)$ instances of the problem, according to all possible candidate sets for $\mathcal{E}^{(2)}$ such that $\sum_{J_j^{(2)}\notin \mathcal{E}^{(2)}}w_{j}^{(2)}\leq A_2$ holds. Each instance is in fact an instance of the $1 \,\big\vert\, \sum w_{j}^{(1)}E_{j}^{(1)} \geq A_1, \, L^{(2)}_{\max} \leq 0  \,\big\vert\,$ problem which includes the $n$ jobs of agent 1 and only the $O(k)$ early jobs of agent 2. In the reduced instance, we need to determine whether there exists a schedule where each job in $\mathcal{E}^{(2)}$ is indeed early (\emph{i.e.}, completed not later than its due date) and $\sum w_{j}^{(1)}E_{j}^{(1)}\geq A_1$. Below, we show how we can solve this problem in polynomial time.

We begin by renumbering the jobs in $\mathcal{E}^{(2)}$ such that $d_1^{(2)}\leq d_2^{(2)}\leq ... \leq d_{k'}^{(2)}$, where $k'=|\mathcal{E}^{(2)}|$. As in Section~\ref{section: EC}, we add two dummy jobs $J_{0}^{(1)}$ and $J_{n+1}^{(1)}$ with $d_{0}^{(1)}=p_{0}^{(1)}=w_{0}^{(1)}=p_{n+1}^{(1)}=w_{n+1}^{(1)}=0$ and $d_{n+1}^{(1)}=\infty$. Furthermore, we again let $k(a,b,\ell)$ denote, for $0 \leq a < b \leq n$ and $0 \leq \ell < k'$, the length of the maximal consecutive subsequence of jobs in $\{J^{(2)}_{\ell+1},\ldots,J^{(2)}_{k'}\}$ with total processing time not greater than $d_b^{(1)}-d_a^{(1)}-p_b^{(1)}$. Then Lemma~\ref{lemma: auxilary_jit2} holds here as well, and we can assume that if $J^{(1)}_a$ and $J^{(1)}_b$ are both scheduled in JIT mode, with no other agent 1 jobs between them, then $J^{(2)}_{\ell+1},\ldots,J^{(2)}_{\ell+k(a,b,\ell)}$ is scheduled right after the completion of $J_a^{(1)}$ (and no other jobs are scheduled prior to the completion of $J_b^{(1)}$).

Since no jobs of agent 2 are allowed to be late, a partial schedule that includes jobs $J_a^{(1)}$ and $J_b^{(1)}$ as two consecutive JIT jobs with $\ell'$ early jobs of agent 2 scheduled before $J_a^{(1)}$ is a feasible partial schedule with $\ell'+ k(a,b,\ell')$ early jobs of agent 2 scheduled before $J_b^{(1)}$ only if the following two conditions holds:
\begin{itemize}
\item[] \textit{Condition 1}: $d_b^{(1)}-d_a^{(1)}-p_b^{(1)}\geq 0$; and
\item[] \textit{Condition 2}: $d_a^{(1)}+\sum_{i=\ell'+1}^{j}p_{i}^{(2)}-d_j^{(2)}\leq 0$ for all $j=\ell'+1,...,l'+\mid k(a,b,\ell')\mid$.
\end{itemize}

Now, let $W(b,\ell)$ represent the maximum total weighted number of JIT jobs among all partial schedules on job set $\{J_1^{(1)},\ldots,J_b^{(1)},J_1^{(2)},\ldots,J_\ell^{(2)}\}$ where all jobs $\{J_1^{(2)},\ldots,J_\ell^{(2)}\}$ are early and $J_b^{(1)}\in \widehat{\mathcal{E}}^{(1)}$ is the last scheduled job. Note that as opposed to Section~\ref{section: EC}, here this value represents the criteria of agent 1. Each value $W(b,\ell)$ can be computed with the following recursion that considers all possible ways of appending partial schedules that have fewer early jobs of agent~2 and one less JIT job of agent~1:
\begin{equation}
\label{eqn: 155}
W(b,\ell)=\max_{\substack{0\leq a \leq b-1 \\
0 \leq \ell' \leq \ell-1}}
\begin{cases}
W(a,\ell')+w_b^{(1)} & : \, \text{conditions 1 and 2 hold and } \ell=\ell'+ |k(a,b,\ell')|.\\
-\infty & :\, \text{otherwise}.
\end{cases}
\end{equation}
The base cases for this recursion are given by $W(0,0) = 0$, and $W(0,\ell) = -\infty$ for $\ell \neq 0$.

Our algorithm reports that there exists a feasible solution to the instance of $1 \,\big\vert\, \sum w_{j}^{(1)}E_{j}^{(1)} \geq A_1, \, \sum w_{j}^{(2)}U_{j}^{(2)} \leq A_2 \,\big\vert\,$ problem iff for some set $\mathcal{E}^{(2)}$ with $\sum_{J_j^{(2)}\notin \mathcal{E}^{(2)}}w_{j}^{(2)}\leq A_2$ we have $W(n+1,|\mathcal{E}^{(2)}|) \geq A_1$. The running time of this algorithm can be bounded by $O(2^kk^2n^2)$, using a similar analysis to the one given in Section~\ref{section: EC}. Thus, we obtain:

\begin{theorem}
\label{theorem: EU1}
The $1 \,\big\vert\, \sum w_{j}^{(1)}E_{j}^{(1)} \geq A_1, \, \sum w_{j}^{(2)}U_{j}^{(2)} \leq A_2 \,\big\vert\,$ problem is solvable in $O(2^kk^2n^2)$ time.
\end{theorem} 

%% file: EE.tex

It is known that the $1 \,\big\vert\, \sum w_{j}^{(1)}E_{j}^{(1)} \geq A_1, \, \sum w_{j}^{(2)}E_{j}^{(2)} \geq A_2 \,\big\vert\,$ problem is NP-complete, and that it is polynomial-time solvable if the weights of either one of the two agents are all equal~\cite{Sha}. Below we show that (i) the problem is NP-complete even for the case of unit processing time; and that (ii) the general problem (with arbitrary processing time) is FPT with respect to $k$.

\begin{theorem}
\label{theorem: single job type}
The $1 \,\big\vert\, p_{j}^{(1)}=p_{j}^{(2)}=1, \sum w_{j}^{(1)}E_{j}^{(1)} \geq A_1, \, \sum w_{j}^{(2)}E_{j}^{(2)} \geq A_2 \,\big\vert\,$ problem is \textnormal{NP}-complete.
\end{theorem}

\begin{proof}
Given an instance $(X,z)$ to the NP-hard \textsc{Partition} problem (see Definition ~\ref{Definition: Partition}), we construct the following instance for the $1 \,\big\vert\, p_{j}^{(1)}=p_{j}^{(2)}=1, \sum w_{j}^{(1)}E_{j}^{(1)} \geq A_1, \, \sum w_{j}^{(2)}E_{j}^{(2)} \geq A_2 \,\big\vert\,$ problem. We set $n=k=m$, and for $j=1,...,n$ we set $w^{(1)}_j =w^{(2)}_j = x_j$ and $d^{(1)}_j =d^{(2)}_j = j$ (recall that here $p^{(1)}_j =p^{(2)}_j = 1$). Moreover, we set $A_1=A_2=z$. Note that since jobs $J^{(1)}_j$ and $J^{(2)}_j$ have the same due date of $j$ for $j=1,...,n$, only one of them can be completed in a JIT mode. Thus, the total gain for both agents is restricted to be not more than $\sum_{j=1}^{n} x_j=2z$.

Suppose that $X$ can be partitioned into two sets $S_1$ and $S_2$ with $\sum_{x_j \in S_1} x_j = \sum_{x_j \in S_2} x_j = z$. Schedule each job in $\{J_{j}^{(i)}\mid x_{j}\in S_{i}\}$ during time interval $(j-1,j]$, and schedule the remaining jobs in an arbitrary order. Then each job in $\{J_{j}^{(i)}\mid x_{j}\in S_{i}\}$ is scheduled in JIT mode, and $\widehat{\mathcal{E}}^{(i)}=\{J_{j}^{(i)}\mid x_{j}\in S_{i}\}$ for $i=1,2$. The fact that $\sum_{x_j \in S_i} x_j = z$, for $i=1,2$, implies that in $\sigma$ we have $\sum_{J_j^{(i)}\in \widehat{\mathcal{E}}^{(i)}} w_j^{(i)}=z$ for each agent $i$. Thus, there is a feasible schedule for our constructed instance.

For the other direction, suppose there exists a schedule $\sigma$ with $\sum w_{j}^{(i)}E_{j}^{(i)} \geq A_i=z$ for $i=1,2$. Then since the total gain in the objective function of both agents is restricted to be not more than $2z$, we have that $\sum w_{j}^{(i)}E_{j}^{(i)}=z$ for $i=1,2$. This means that by setting $S_i=\{x_j:E_j^{(i)}=1\}$, for $i=1,2$, we obtain a solution for $(X,z)$ with $\sum_{x_j \in S_i} x_j = z$ for $i=1,2$. \qed
\end{proof}

Next, we show that the $1 \,\big\vert\, \sum w_{j}^{(1)}E_{j}^{(1)} \geq A_1, \, \sum w_{j}^{(2)}E_{j}^{(2)} \geq A_2 \,\big\vert\,$ problem is FPT with respect to $k$. To this end, it is convenient to view all jobs in $\mathcal{J}^{(1)} \cup \mathcal{J}^{(2)}$ as time intervals. For a job $J_j^{(i)} \in \mathcal{J}^{(1)} \cup \mathcal{J}^{(2)}$, let the \emph{time interval} of $J_j^{(i)}$ be $I_j^{(i)}=(d_j^{(i)}-p_j^{(i)},d_j^{(i)}]$. Then if $J_j^{(i)}$ is required to be scheduled in JIT mode, it has to be scheduled within its time interval $I_j^{(i)}$. Thus, any pair of jobs $J_{j_1}, J_{j_2} \in \mathcal{J}^{(1)} \cup \mathcal{J}^{(2)}$ can be simultaneously scheduled in JIT mode iff $I_{j_1} \cap I_{j_2} = \emptyset$. This means that our goal now translates to finding a set of pairwise disjoint time intervals, for which the total weight of set of intervals related to the jobs of each agent met its bound.

We try out all possible candidates for $\widehat{\mathcal{E}}^{(2)}$; that is, all subsets of agent~2 jobs $\mathcal{J}_2 \subseteq \mathcal{J}^{(2)}$ that have pairwise disjoint time intervals and $\sum_{J^{(2)}_j \in \mathcal{J}_2} w^{(2)}_j \geq A_2$. For each such subset $\mathcal{J}_2$, we compute the subset of agent~1 jobs $\mathcal{J}_1 \subseteq \mathcal{J}^{(1)}$ with time intervals that do not intersect any time interval of a job in $\mathcal{J}_2$. That is, $\mathcal{J}_1 = \{J_j^{(1)} \in  \mathcal{J}^{(1)} : I_j^{(1)} \cap I_i^{(2)} = \emptyset \text{ for all } J_i^{(2)} \in \mathcal{J}_2\}$. We then compute a maximum weight pairwise disjoint subset $\mathcal{I}^*_1$ of intervals in the set $\mathcal{I}_1 = \{ I^{(1)}_j : J^{(1)}_j \in \mathcal{J}_1\}$, and report that we have found a feasible schedule if the total weight of the jobs in $\mathcal{J}^*_1 = \{ J^{(1)}_j : I^{(1)}_j \in \mathcal{I}^*_1\}$ is at least $A_1$. If no such subset of jobs $\mathcal{J}^*_1$ is found for any possible candidate for $\widehat{\mathcal{E}}^{(2)}$, we report that our instance has no feasible schedule.

Correctness of our algorithm follows from the fact that we compute the optimal set $\widehat{\mathcal{E}}^{(1)}$ for each possible candidate for $\widehat{\mathcal{E}}^{(2)}$. There are $O(2^k)$ candidates $\mathcal{J}_2$ for $\widehat{\mathcal{E}}^{(2)}$. For each candidate $\mathcal{J}_2$, computing the set $\mathcal{J}_1$ can be done in $O(n \log n + k \log n)= O(n \log n)$ time. Moreover, the set $\mathcal{J}^*_1$ can be computed in $O(n \log n)$ time (\emph{e.g.} using~\cite{GuptaLL82}). Thus, in total, our algorithm runs in $O(2^k n \log n )$ time.

\begin{theorem}
\label{theorem: single job bob4}
The $1 \,\big\vert\, \sum w_{j}^{(1)}E_{j}^{(1)} \geq A_1, \, \sum w_{j}^{(2)}E_{j}^{(2)} \geq A_2 \,\big\vert\,$ problem is solvable in $O(2^k n \log n )$ time.
\end{theorem}

%% file: conclusion.tex

In this paper we initiated a parameterized analysis for two-agent single-scheduling problems, where the parameter studied is the number $k$ of jobs belonging to the second agent. We considered three possible scheduling criteria -- total weighted completion time, total weighted number of tardy jobs, and total weighted number of JIT jobs -- and all possible combinations of these criteria for each agent. Our analysis shows that parameter $k$ indeed provides various positive results in different settings, and is summarized in Table~\ref{table: results} below.

\begin{table}[h!]
\begin{center}
\begin{tabular}{|c|c|c|c|}
\hline
& & &\\
& \quad \boldmath{$\sum w^{(2)}_j C^{(2)}_j \leq A_2$} \quad & \quad \boldmath{$\sum w^{(2)}_j U^{(2)}_j \leq A_2$} \quad & \quad \boldmath{$\sum w^{(2)}_j E^{(2)}_j \geq A_2$} \quad \\
& & &\\
\hline
& Hard for $w_j^{(2)}=1$ (Th.~\ref{theorem: single job bob}), & Hard for $w_j^{(2)}=1$ (Th.~\ref{theorem: single job bob}), & Hard even when \\
\quad \boldmath{$\sum w^{(1)}_j C^{(1)}_j \leq A_1$} \quad & \quad FPT for $w_j^{(1)}=1$ (Th.~\ref{theorem: unit weights2}), \quad & \quad FPT for $w_j^{(1)}=1$ (Th.~\ref{theorem: unit processing}). \quad & \quad $w_j^{(i)}=1$ (Cor.~\ref{theorem: hardness1}). \quad \\
& \quad FPT for $p_j^{(1)}=1$ (Th.~\ref{theorem: unit processing}). \quad & &  \\
\hline
&  Hard in general (Cor.~\ref{thm: SigmaWUhardness}), & Hard in general (Cor.~\ref{thm: SigmaWUhardness}), & Hard even when \\
\quad \boldmath{$\sum w^{(1)}_j U^{(1)}_j \leq A_1$} \quad & Open for $w_j^{(1)}=1$. & FPT for $w_j^{(1)}=1$ (Th.~\ref{theorem: uuu}), & \quad $w_j^{(i)}=1$ and $d^{(1)}_j=d$ (Cor.~\ref{theorem: hardness2}). \quad \\
& & \quad FPT for $p_j^{(i)}=1$ (Th.~\ref{theorem: uuu2}). \quad  & \\
\hline
& Open in general, &  & \\
\quad \boldmath{$\sum w^{(1)}_j E^{(1)}_j \geq A_1$} \quad & FPT when $w_j^{(1)}=1$ (Th.~\ref{theorem: EC1}). & \LARGE{FPT} \normalsize(Th.~\ref{theorem: EU1}) & \LARGE{FPT} \normalsize(Th.~\ref{theorem: single job bob4})\\
& & &\\
\hline
\end{tabular}
\end{center}
\caption{\label{table: results} A table depicting most of the results presented in the paper. Rows correspond to the different scheduling criteria for agent 1, while columns are associated with the different criteria for agent 2.}
\end{table}

There are several directions to directly extend our work. First, one can find different parameters such as the number of different processing times in the input, or the number of different due dates. Many such parameterizations make perfect sense for practical applications. Second, one consider other scheduling criteria not considered in this paper such as the maximal lateness. Below we list the three most important questions that were left open directly from our work:
\begin{enumerate}
\item Determine the parameterized complexity of the $1 \,\big\vert\, \sum U_{j}^{(1)} \leq A_1, \, \sum w_{j}^{(2)} C_{j}^{(2)} \leq A_2 \,\big\vert\,$ problem, or any of its variants (unit weights, unit processing times, etc ...).
\item Determine the parameterized complexity of the $1 \,\big\vert\, \sum w_{j}^{(1)} E_{j}^{(1)} \geq A_1, \, \sum w_{j}^{(2)} C_{j}^{(2)} \leq A_2 \,\big\vert\,$ problem.
\item Can the running times of the algorithms presented in Sections~\ref{section: CC2} and~\ref{section: CC3} be improved? More specifically, are there purely combinatorial algorithms for these problems, avoiding ILPs altogether?
\end{enumerate} 

%% file: Acknowledgments.tex
The research leading to these results has received funding from the People Programme (Marie Curie Actions) of the European Union's Seventh Framework Programme (FP7/2007-2013) under REA grant agreement number 631163.11, and by the Israel Science Foundation (grant No. 1055/14). Nimrod Talmon was supported by a postdoctoral fellowship from I-CORE ALGO.